\documentclass[a4paper,USenglish,cleveref, autoref, thm-restate]{lipics-v2019}

\usepackage{tikz}
\usepackage{booktabs}

\tikzstyle{intv}=[fill=solid_silver!25,draw=solid_silver!25,opacity=0.5]

\tikzstyle{rot}=[color=ptv_red,ultra thick]
\tikzstyle{blau}=[color=ptv_blue,ultra thick]
\tikzstyle{gruen}=[color=ptv_green,ultra thick]
\tikzstyle{grau}=[color=solid_silver,ultra thick]

\definecolor{ptv_red}{RGB}{224,33,41}
\definecolor{solid_silver}{RGB}{91,91,90}
\definecolor{bright_silver}{RGB}{188,189,192}
\definecolor{ptv_blue}{RGB}{0,172,205}

\title{Efficient Route Planning with Temporary Driving Bans, Road Closures, and Rated Parking Areas} 

\titlerunning{Efficient Route Planning with Temporary Driving Restrictions}

\author{Alexander Kleff}{PTV Group, Karlsruhe, Germany}{alexander.kleff@ptvgroup.com}{https://orcid.org/0000-0001-5609-944X}{}

\author{Frank Schulz}{PTV Group, Karlsruhe, Germany}{frank.schulz@ptvgroup.com}{https://orcid.org/0000-0002-3144-8479}{}

\author{Jakob Wagenblatt}{Karlsruhe Institute of Technology, Germany}{jakob.wagenblatt@student.kit.edu}{https://orcid.org/0000-0001-7045-0566}{}

\author{Tim Zeitz}{Karlsruhe Institute of Technology, Germany}{tim.zeitz@kit.edu}{https://orcid.org/0000-0003-4746-3582}{}

\authorrunning{A. Kleff, F. Schulz, J. Wagenblatt, and T. Zeitz}

\Copyright{Alexander Kleff, Frank Schulz, Jakob Wagenblatt, and Tim Zeitz}

\ccsdesc[500]{Theory of computation~Shortest paths}
\ccsdesc[300]{Mathematics of computing~Graph algorithms}
\ccsdesc[500]{Applied computing~Transportation}

\keywords{driving bans, realistic road networks, route planning, shortest paths}

\nolinenumbers 

\hideLIPIcs  

\EventEditors{Simone Faro and Domenico Cantone}
\EventNoEds{2}
\EventLongTitle{18th International Symposium on Experimental Algorithms (SEA 2020)}
\EventShortTitle{SEA 2020}
\EventAcronym{SEA}
\EventYear{2020}
\EventDate{June 16--18, 2020}
\EventLocation{Catania, Italy}
\EventLogo{}
\SeriesVolume{160}
\ArticleNo{14}

\DeclareMathAlphabet{\mathcalligra}{T1}{calligra}{m}{n}

\newcommand*{\NP}{\textsl{NP}}

\newcommand*{\PI}{\mathbb{N}}
\newcommand*{\NNI}{\mathbb{N}_0}
\newcommand*{\RI}{\mathbb{Q}_{\ge 0}}
\newcommand*{\RIU}{\mathbb{Q}_{\ge 0}\cup \{\infty\}}

\newcommand*{\C}{\mathcal{C}}
\newcommand*{\PH}{H}
\newcommand*{\T}{\mathcal{T}}

\newcommand*{\source}{\ensuremath{s}}
\newcommand*{\target}{\ensuremath{z}}
\newcommand*{\atime}{t}
\newcommand*{\aperiod}{p}

\newcommand*{\TW}{\Phi}
\newcommand*{\vertices}{V}
\newcommand*{\vertexa}{u}
\newcommand*{\vertexb}{v}
\newcommand*{\edges}{E}
\newcommand*{\edge}{e}

\newcommand*{\drive}{\delta}
\newcommand*{\maxCat}{r}
\newcommand*{\categ}{\rho}

\newcommand*{\cdriv}{d}
\newcommand*{\cwait}[1][0]{w_{#1}}

\newcommand*{\pth}{R}
\newcommand*{\arrival}{A}
\newcommand*{\depart}{D}
\newcommand*{\len}{\ell}

\newcommand*{\Hmin}{t^{min}}
\newcommand*{\Hmax}{t^{max}}

\newcommand*{\twa}{t^{closed}}
\newcommand*{\two}{t^{open}}

\newcommand*{\update}{t^{visit}}

\begin{document}

\maketitle

\begin{abstract}
We study the problem of planning routes in road networks when certain streets or areas are closed at certain times.
For heavy vehicles, such areas may be very large since many European countries impose temporary driving bans during the night or on weekends.
In this setting, feasible routes may require waiting at parking areas,
and several feasible routes with different trade-offs between waiting and driving detours around closed areas may exist. 
We propose a novel model in which
driving and waiting are assigned abstract costs, and waiting costs are location-dependent to reflect the different quality of the parking areas.
Our goal is to find Pareto-optimal routes with regards to arrival time at the destination and total cost.
We investigate the complexity of the model and determine a necessary constraint on the cost parameters such that the problem is solvable in polynomial time.
We present a thoroughly engineered implementation and perform experiments on a production-grade real world data set.
The experiments show that our implementation can answer realistic queries in around a second or less which makes it feasible for practical application.

\end{abstract}

\section{Introduction}
\label{sec:intro}

Many European countries impose temporary driving bans for heavy vehicles.
Driving may be restricted during the night, on weekends, and on public holidays.
Such bans may apply to the whole road network of a country or parts of it.
When routing a heavy vehicle from a source to a destination, it is crucial to take these temporary driving bans into account.
But it is not only about heavy vehicles.
Temporary closures of bridges, tunnels, border crossings, mountain pass roads, or certain inner-city areas as well as closures due to roadworks may affect all road users alike.
In case of road space rationing in cities, the driving restriction may depend on the license plate number.
To sum up, temporary driving restrictions exist in different forms, and the closing and re-opening times of a road segment must be considered in the route planning.

As a consequence of temporary driving restrictions, waiting times may be inevitable and even last for hours.
During such waiting hours, the vehicle must be parked properly, and thus a suitable parking area has to be found.
The driving time of the detour from and to such a parking area should also be incorporated in the route planning.
Unfortunately, the underlying shortest (here: quickest) path problem becomes \NP-hard if waiting is only allowed at dedicated locations~\cite{or-tnp-89}.
This is because in this case, the so-called \emph{FIFO} (first in, first out) property is not satisfied, that is, the property that a driver cannot arrive earlier by departing later.
Thus, our first research question is how we can consider dedicated waiting locations without making the underlying problem \NP-hard.
It is our aim to obtain a feasible running time even for long-distance routes.

In practice, we often find that small parking areas without any facilities like public toilets or restaurants cause the least detour.
So an algorithm that looks for the shortest route, that is, a route with the shortest driving time, would select small parking areas in these cases, provided that waiting is necessary.
But the longer the waiting time is, the more vital a secure and pleasant place for waiting becomes.
So it may be important for the driver that nearby facilities of the parking area and their quality are somehow taken into account as well.
How to do this is our second research question.

In our setting, a single-criterion objective is not practical.
A driver may not always be in favor of the shortest route if that means to spend a very long time waiting and to arrive at the destination considerably later than on the quickest route, that is, a route with the earliest arrival at the destination.
Conversely, a driver may not always be interested in a quickest route if that route means to take an unjustified long detour around temporarily closed road segments that could be avoided by waiting in a comfortable place.
In other words, an early arrival at the destination (and thus low opportunity costs), little driving time (and thus low fuel costs), and pleasant waiting conditions (and thus high driver satisfaction) are competing criteria.
Solutions can differ significantly with regards to these criteria.
How to deal with this and find reasonable routes is the third research question.

In this paper, we answer these questions as follows:
\begin{enumerate}
	\item We present a model in which waiting is allowed at any vertex and any edge at any time in the road graph but waiting on edges and waiting on those vertices that do not correspond to parking areas is penalized.
	This is done by assigning a cost to time spent waiting there.
	Since driving comes at a price, too, we also assign a cost per time unit spent driving.
	As we will show, we can find a route with least costs in polynomial time if both cost parameters are set to the same value.
	\item We assume that the nearby facilities of a parking area and their quality can be expressed by some single rating number.
	To take account of this, we assign a waiting cost to every corresponding vertex as well.
	This cost is lower than the cost of waiting anywhere else in the road graph, and it is even lower the higher the rating of the parking area is.
	\item We return routes that are Pareto-optimal with regards to arrival time at the destination on the one hand and total costs on the other.
	Despite the potentially larger output, our algorithm still runs in polynomial time under the same condition as before.
\end{enumerate}

As our experiments reveal, many queries within Europe are answered within milliseconds.
Except some pathological cases, even more complex queries with four or more Pareto-optimal solutions are solved in less than a second.

\subparagraph{Related Work}

Many route planning problems are modeled as shortest path problems.
To this day, the theoretically fastest known algorithm to find shortest paths on graphs with static non-negative edge weights is the algorithm of Dijkstra~\cite{d-ntpcg-59}.
However, for many practical applications, it is not fast enough.
One approach to speed up the computation is to reduce the search space of Dijkstra's algorithm by guiding the search towards the destination by means of estimates of the remaining distance to the destination.
It is known as the A* algorithm~\cite{hnr-afbhd-68}.
Since the advent of routing services, a lot of research has been done on efficient algorithms for routing in road networks.
Routing services have to answer many queries on the same network.
This can be used to speed up shortest path queries through precomputed auxiliary data.
Many approaches exploit certain characteristics of road networks, for example the hierarchical structure (freeways are more important than rural roads).
For an extensive overview, we refer to~\cite{bdgmpsww-rptn-16}.
One particularly popular speed-up technique are Contraction Hierarchies~\cite{gssv-erlrn-12}.
During preprocessing, additional shortcut edges are inserted into the graph, which skip over unimportant vertices.
This preprocessing typically takes a few minutes.
Then, shortest path queries can be answered in less than a millisecond.

A natural approach to handle driving restrictions is to model them as time-dependent travel times~\cite{d-aassp-69}.
For the blocked time, the travel time of the edge can be set to infinity.
Time-dependent route planning has also received some attention and effective speed-up techniques are known~\cite{bgsv-mtdtt-13,bdpw-dtdrp-16,dn-crdtd-12,d-tdsr-11,ndls-bastd-12}.

Variants of our problem have been studied in the literature.
In~\cite{desaulniers2000shortest} a related problem is discussed where nodes (not edges) have time windows and waiting is associated with a cost.
In~\cite{pugliese2013survey} an overview is given over different exact approaches to solving shortest path problems with resource constraints.
Time windows on nodes are a specific kind of constraint in this framework.
More specialized models for routing applications have been proposed.
The authors of~\cite{twb-rpbtd-18} study the problem of planning a single break, considering driving restrictions and provisions on driver breaks.
They aim to find only the route with the earliest arrival.

\subparagraph{Contribution}

We present a novel model that helps answer our three research questions in the context of temporary driving restrictions and dedicated waiting locations.
To the best of our knowledge, this is the first unifying approach that gives answers to all three research questions.
Our theoretical analysis reveals that our model can be solved to optimality in polynomial time, given certain restrictions on the parameterization.
The experimental evaluation of our implementation demonstrates a practical running time.

\subparagraph{Outline}

In \cref{sec:problem}, we give a formal definition of the routing problem at hand.
In \cref{sec:algorithm}, we present an exact algorithm for this problem.
In \cref{sec:analysis}, we analyze the complexity of the problem and show that our algorithm runs in polynomial time if the costs for driving are the same as for waiting anywhere else than at a dedicated waiting location.
In \cref{sec:impl}, we describe techniques to speed-up the computation.
In \cref{sec:exp}, we present the main results of our experiments.
Finally, we conclude in \cref{sec:conclusion}.

\section{Problem}
\label{sec:problem}

A problem instance comprises a \emph{road graph with ban intervals on edges, driving costs and location-dependent waiting costs} (or \emph{road graph with ban intervals and costs} for short) as well as a set of \emph{queries}.
The road graph is characterized by the following attributes:

\begin{itemize}
	\item A set $\vertices$ of $n$ vertices and a set $\edges$ of $m$ directed edges.
	\item A mapping $\TW$ that maps each edge $\edge \in \edges$ to a sequence of disjoint time intervals, where the edge is considered to be \emph{closed} during each interval.
	Precisely, for any \emph{ban interval} $[\twa,\two) \in \TW(\edge)$ of an edge $\edge$, $\two$ denotes the first point in time after $\twa$ where the edge is open again.
	Here and in the following, all points in time are integers and the length of an interval is denoted by $|[\twa,\two)|$ and equals $\two - \twa > 0$.
	During such a time span, a vehicle on the corresponding road segment must not move.
  We denote the total number of ban intervals as $b$.
	\item A mapping $\drive: \edges \rightarrow \PI$ that maps each edge $\edge:=(u,v)\in \edges$ to the time $\drive(\edge)$ that it takes to drive from $u$ to $v$, provided the edge is open.
	\item A mapping $\categ$ that maps each vertex to a rating in $\{0,1,\ldots,\maxCat\}$ with $\maxCat \le n$.
	Rating 0 means \emph{unrated}, that is, it is assumed that it is highly difficult, dangerous, and not allowed to park the vehicle there.
	In contrast to an unrated location, we call a vertex $v$ with $\categ(v) > 0$ a \emph{parking location}.
	\item A parameter set of abstract costs, consisting of $\cdriv \in \RI$, the cost per unit of driving time, and $\cwait[i] \in \RI$ for all $i$ from 0 to $\maxCat$, the cost per time unit of waiting on a vertex with rating~$i$.
	Edges are always unrated so waiting there costs $\cwait[0]$ per time unit.
	W.l.o.g. $\cwait[i] < \cwait[i-1]$ holds for all $i$ between 1 and $\maxCat$, that is, we assume that waiting on vertices with a higher rating costs less than waiting on those with a lower rating.
\end{itemize}

A \emph{$\vertexa$-$\vertexb$-route} is a triple $(\pth, \arrival,\depart)$ of three sequences of the same length $\len:=|\pth|=|\arrival|=|\depart|$.
Here, $\pth$ is the sequence of vertices along the route.
It describes a (not necessarily simple) \emph{path} in the graph that starts at $\vertexa$ and ends in $\vertexb$, that is,
$\edge_i:=(\pth[i],\pth[i+1])\in \edges$ for all $1 \le i< \len$ and $\pth[1]=\vertexa$ and $\pth[\len]=\vertexb$.
The other two sequences $\arrival$ and $\depart$ denote the \emph{arrival times} and the \emph{departure times} from the respective vertices,
where $\arrival[i] \le \depart[i]$ for all $1\le i \le \ell$ and $\arrival[{i+1}] - \depart[i] \ge \drive(\edge_i)$ for all $1\le i < \ell$ holds.

A query comprises a \emph{source} $\source \in \vertices$ and a \emph{destination} $\target \in \vertices$ as well as a \emph{planning horizon} $\PH$.
The latter is defined as the time interval between an \emph{earliest departure time} $\Hmin$ from $\source$ and a \emph{latest arrival time} $\Hmax$ at $\target$.
	Waiting costs arise as soon as the planning horizon opens.
For a given query, we look for \emph{feasible} $\source$-$\target$-routes.
A route is feasible with respect to the planning horizon if $\arrival[1]=\Hmin$ and $\depart[\len]\le \Hmax$.
In addition, ban intervals must be taken account of.
Let $T_i:=[\depart[{i}],\arrival[{i+1}])$ be the time interval in which the edge $\edge_i:=(\pth[i],\pth[{i+1}])$ of the route's path is traversed.
A route is feasible with respect to the ban intervals if $\sum_{I \in \TW(\edge_i)} | T_i \cap I | \le |T_i| - \drive(\edge_i)$ for all $1\le i < \ell$.
Here, $\sum_{I \in \TW(\edge_i)} | T_i \cap I |$ is the time during which the edge between $\pth[i]$ and $\pth[{i+1}]$ is closed while the edge is being traversed.

Let \emph{travel time} include driving time and waiting time.
The \emph{travel time costs} of a route are the sum of the waiting time costs and the driving time costs.
So given a route of length $\ell$, the travel time costs are
	\[\sum_{i=1}^{\ell} \cwait[{\categ(\pth[i])}] \cdot \left(\depart[i] - \arrival[i]\right) + \sum_{i=1}^{\ell-1} \cwait \cdot \left(\arrival[{i+1}] - \depart[i] - \drive(\edge_i)\right) + \cdriv \cdot \drive(\edge_i),
\]
where we use $\edge_i:=(\pth[i],\pth[{i+1}])$.
We say an $\source$-$\target$-route is \emph{Pareto-optimal} (or simply \emph{optimal}) if it is feasible and if its travel time costs are less or its arrival time at $\target$ is earlier or equality holds in both cases compared to any other feasible $\source$-$\target$-route.
For a query, the objective is to find a maximal set of (Pareto-)optimal $\source$-$\target$-routes such that no two routes in the set have both the same arrival time at~$\target$ and the same travel time costs.

\section{Algorithm}
\label{sec:algorithm}

The algorithm maintains a priority queue.
Each entry of the queue consists of a vertex and a point in time within the planning horizon as key.
We say a vertex is \emph{visited} at a certain point in time whenever we remove the top entry from the queue, that is, an entry with the earliest time among the entries in the queue.
At every vertex $v \in \vertices$, we store a time-dependent function $\C_v: \PH \rightarrow \RIU$.
It maps a point in time $\atime$ within the planning horizon $\PH$ to an upper bound on the minimum travel time cost over all $\source$-$v$-routes that end in $v$ at time $\atime$.
We call this function \emph{cost profile} of $v$ or, more general, \emph{label} of $v$.
The algorithm works in a \emph{label correcting} manner in the sense that a vertex may be visited multiple times, albeit at different times within the planning horizon.

Before we describe the phases of the algorithm in greater detail, we introduce an auxiliary time-dependent function~$\T_\edge$ for every edge $\edge \in \edges$.
It maps a time $\atime$ at the head $\vertexb$ of an edge $\edge:=(\vertexa,\vertexb)$ to the \emph{shortest travel time} that it takes to traverse the edge from $\vertexa$ to $\vertexb$ completely and be at $\vertexb$ at time $\atime$, possibly including waiting time.
That is, for a time $\atime$ at $\vertexb$, $\T_\edge(\atime)$ is the minimum period~$\aperiod$ such that $\aperiod - \sum_{I \in \TW(\edge)} | [\atime-\aperiod,\atime) \cap I | \ge \drive(\edge)$ holds if such a $\aperiod$ exists, and $\infty$ otherwise.
In other words, $\atime - \T_\edge(\atime)$ is the latest departure time from $\vertexa$ in order not to arrive at $\vertexb$ later than at time $\atime$.
An example is given in \cref{fig:exampleTravelTimeFunction}.

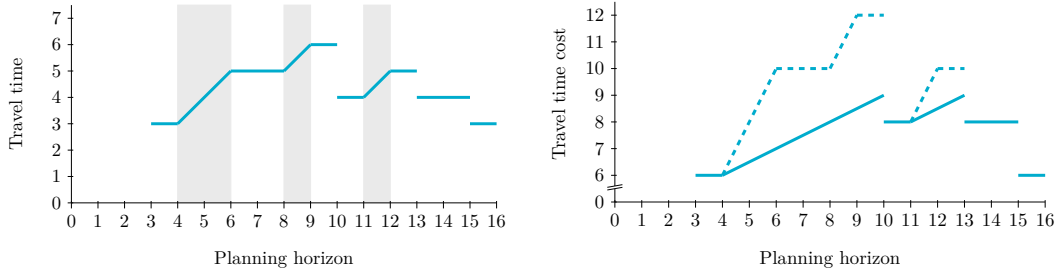
\begin{figure}
     \begin{subfigure}[b]{0.49\textwidth}
          \centering
          \resizebox{\linewidth}{!}{\begin{tikzpicture}[y=.5cm, x=.5cm]

    \draw[intv] (4,7.5) rectangle (6,0);
    \draw[intv] (8,7.5) rectangle (9,0);
    \draw[intv] (11,7.5) rectangle (12,0);

    \draw (0,0) -- coordinate (x axis mid) (16,0);
    \draw (0,0) -- coordinate (y axis mid) (0,7.5);
    
    \foreach \x in {0,...,16}
        \draw (\x,1pt) -- (\x,-3pt)
        node[anchor=north] {\x};

    \foreach \y in {0,...,7}
        \draw (1pt,\y) -- (-2pt,\y)
        node[anchor=east] {\y};

    \node[below=0.8cm] at (x axis mid) {Planning horizon};
    \node[rotate=90, above=0.8cm] at (y axis mid) {Travel time};

	\draw[blau] (3,3) -- (4,3);
	\draw[blau] (4,3) -- (6,5);
	\draw[blau] (6,5) -- (8,5);
	\draw[blau] (8,5) -- (9,6);
	\draw[blau] (9,6) -- (10,6);
	\draw[blau] (10,4) -- (11,4);
	\draw[blau] (11,4) -- (12,5);
	\draw[blau] (12,5) -- (13,5);
	\draw[blau] (13,4) -- (15,4);
	\draw[blau] (15,3) -- (16,3);

\end{tikzpicture}  }
          \caption{Travel time function $\T_\edge$ of an edge $\edge$ with ban intervals (grey) and a driving time $\drive(\edge)$ of 3. The latest departure to be at $\vertexb$ at time $t$ is $\atime-\T_\edge(\atime)$.}
          \label{fig:exampleTravelTimeFunction}
     \end{subfigure}
     \hspace{1mm}
     \begin{subfigure}[b]{0.49\textwidth}
          \centering
          \resizebox{\linewidth}{!}{\begin{tikzpicture}[y=.5cm, x=.5cm]

    \draw (0,0) -- coordinate (x axis mid) (16,0);

    \draw (0,0) -- coordinate (y axis mid) (0,0.5);
    \draw (0,0.6) -- coordinate (y axis mid) (0,7.5);
    \draw (-4pt,0.45) -- (4pt,0.55);
    \draw (-4pt,0.55) -- (4pt,0.65);

    \foreach \x in {0,...,16}
        \draw (\x,1pt) -- (\x,-3pt)
        node[anchor=north] {\x};

    \draw (1pt,0) -- (-2pt,0)
    node[anchor=east] {0};

    \foreach \y in {6,...,12}
        \draw (1pt,\y-5) -- (-2pt,\y-5)
        node[anchor=east] {\y};

    \node[below=0.8cm] at (x axis mid) {Planning horizon};
    \node[rotate=90, above=0.8cm] at (y axis mid) {Travel time cost};

	\draw[blau] (3,1) -- (4,1);
	\draw[blau, dashed] (4,1) -- (6,5);
	\draw[blau] (4,1) -- (10,4);
	\draw[blau, dashed] (6,5) -- (8,5);
	\draw[blau, dashed] (8,5) -- (9,7);
	\draw[blau, dashed] (9,7) -- (10,7);
	\draw[blau] (10,3) -- (11,3);
	\draw[blau] (11,3) -- (13,4);
	\draw[blau, dashed] (11,3) -- (12,5);
	\draw[blau, dashed] (12,5) -- (13,5);
	\draw[blau] (13,3) -- (15,3);
	\draw[blau] (15,1) -- (16,1);

\end{tikzpicture}}
          \caption{Cost profile of vertex $\vertexb$ after linking, that is, after considering travel time (dashed) and waiting time at $\vertexb$ (solid).}
          \label{fig:exampleCostProfileNeighbor}
     \end{subfigure}
		\caption{Computing the cost profile of a vertex $\vertexb$.
		Let $\vertexb$ be adjacent to the source $\source$ via an edge $\edge:=(\source,\vertexb)$ with three ban intervals and a driving time $\drive(\edge)$ of 3.
		The corresponding travel time function is given in \cref{fig:exampleTravelTimeFunction}.
		It is infinite between $0=\Hmin$ and $3=\drive(\edge)$.
		In \cref{fig:exampleCostProfileNeighbor}, we see the cost profile $\C_{\vertexb}$ after considering the travel time along the edge (dashed) and after considering waiting at $\vertexb$ (solid).
		Here, the assumed cost parameters are $\cwait[{\categ(\source)}] = 0$, $\cwait[{\categ(\vertexb)}] = 0.5$, and $\cdriv = \cwait = 2$, where $\cwait[{\categ(\source)}] = 0$ implies that the cost profile $\C_{\source}$ at the source is 0 over the whole planning horizon.
		}
		\label{fig:exampleAlgorithm}
\end{figure}

In the initialization phase of the algorithm, we set $\C_\source(t) := \cwait[{\categ(\source)}] \cdot (t-\Hmin)$ for all $t \in \PH$.
For every other $v\in \vertices \setminus \{\source\}$, we set $\C_v(t) := \infty$ for all $t \in \PH$.
Furthermore, we insert the source $\source$ with key $\Hmin$ into the priority queue.

As long as the queue is not empty, we are in the main loop of the algorithm.
In every iteration of the main loop, we remove the top entry from the queue.
Let us suppose we visit a vertex $\vertexa$ at time $\update \ge \Hmin$.
Then, we check for every edge $e:=(\vertexa,\vertexb)$ going out of $\vertexa$ whether we can improve the cost profile $\C_\vertexb$ of $\vertexb$.
We do so in three steps.
In the first step, we consider the travel time along the edge and set
\begin{equation}
    \label{eqn:step1}
    \C'_v(t) := \C_u(t - \T_e(t)) + \cdriv \cdot \drive(e) + \cwait \cdot (\T_e(t) - \drive(e))
\end{equation}
for all $t$ with $\update + \T_e(t) \le t \le \Hmax$.
For all other $t\in \PH$ we set $\C'_v(t) := \infty$.
In the second step, we consider waiting at $v$ at cost $\cwait[{\categ(v)}]$ per time unit and set
\begin{equation}
    \label{eqn:step2}
    \C'_v(t) := \min \{\C'_v(t') + \cwait[{\categ(v)}] \cdot (t-t') \mid \Hmin \le t'\le t \}
\end{equation}
for all $t \in \PH$.
An example of the first two steps is illustrated in \cref{fig:exampleCostProfileNeighbor}.
Finally, in the third step, we compare $\C'_v$ and $\C_v$.
Let $t^*$ be the earliest point in time such that $\C'_v(t^*)$ is less than $\C_v(t^*)$ if such a time $t^*$ exists.
Only if it exists, we set $\C_v(t)$ to the minimum of $\C_v(t)$ and $\C'_v(t)$ for all $t^* \le t \le \Hmax$.
Furthermore, we insert vertex $v$ with key $t^*$ into the priority queue or decrease the key if $v$ is already contained.

When the priority queue is empty, we enter the finalization phase of the algorithm.
We say a time-cost-pair $(t,\C_\target(t))$ with $t \in \PH$ and $\C_\target(t) < \infty$ is Pareto-optimal if there is no time~$t'$ with $\Hmin \le t' < t$ and $\C_\target(t') \le \C_\target(t)$.
In the finalization phase, we extract an $\source$-$\target$-route for every Pareto-optimal time-cost-pair.
So let such a time-cost-pair $(t,\C_\target(t))$ be given.
In order to find a corresponding route $(\pth, \arrival,\depart)$, we initially push $\target$ and $t$ and $t$ to the front of the (empty) sequences $\pth$ and $\arrival$ and $\depart$, respectively.
The following is done iteratively until we reach the source, that is, $\pth[1]=\source$ holds.
First, we look for an incoming edge $e:=(u,\pth[1])$ of $\pth[1]$ and a departure time $t$ from $u$ with
	\[\C_u(t) + \cdriv \cdot \drive(e) + \cwait \cdot (\T_e(\arrival[1]) - \drive(e)) = \C_{\pth[1]}(\arrival[1])
	\]
which must exist.
We push $u$ and $t$ to the front of $\pth$ and $\depart$, respectively.
Then, we push the earliest time $t\le \depart[1]$ such that
	\[\C_{\pth[1]}(t) + \cwait[{\categ({\pth[1]})}] \cdot (\depart[1]-t) = \C_{\pth[1]}(\depart[1])
	\]
holds to the front of the arrival time sequence $\arrival$, and continue with the next iteration.
This concludes the description of the finalization phase and thus the whole algorithm.

For the correctness of the algorithm it is important that the upper bound $\C_\vertexb(t)$ on the minimum travel time cost is tight for all $t \le \update$ and all $\vertexb \in \vertices$ whenever we visit a vertex at time $\update$.
After the main loop, it is tight for every $t\in \PH$ and all $\vertexb \in \vertices$, especially for $\target$.
This can be proven by induction on the time of visit.
The time of visiting a vertex increases monotonically because whenever a vertex is inserted into the queue or its key is decreased, the (new) value of that key can only be later than the current time of visit.

\section{Analysis}
\label{sec:analysis}

In this section, we first show the intractability of the general problem.
Then, we restrict the problem by requiring the driving cost $\cdriv$ to be equal to the
unrated waiting cost $\cwait$, and prove that our algorithm solves the restricted problem
in polynomial time.

\subparagraph{Intractability of the General Problem}

The first two theorems show the intractability of the general problem if $\cdriv \ne \cwait$.
Parking locations are not used in the proofs, so already the simplified problem without
parking locations is intractable if $\cdriv \ne \cwait$.

\begin{restatable}{theorem}{nphard}
\label{thm:np_hard}
If $\cdriv < \cwait[0]$ then it is $\NP$-complete to decide whether there is a feasible route
with travel time costs less than or equal to a given threshold $k$.
\end{restatable}

We prove this theorem in Appendix~\ref{sec:appendixA} by reduction from PARTITION.

\begin{restatable}{theorem}{exproutes}
\label{thm:exp_routes}
If $\cdriv > \cwait[0]$  then the number of Pareto-optimal routes can be exponential in the number of vertices.
\end{restatable}

Given a number of vertices a graph with ban intervals can be constructed that
has exponentially many routes which are all Pareto-optimal. The construction
and the proof that all those routes are Pareto-optimal can be found in
Appendix~\ref{sec:appendixB}.

\subparagraph{Tractable Problem Variant}

For the remaining analysis we assume $\cdriv = \cwait$.
In the setting without parking locations,
there is only one optimal solution, since the quickest solution has also the least cost. Hence, this
setting is a single-criterion shortest path problem with time-dependent edge weights that fulfill
the \emph{FIFO} property and can be solved in polynomial time with a time-dependent variant of
Dijkstra's algorithm~\cite{d-aassp-69}, and also our algorithm reduces to
such a time-dependent Dijkstra variant and has polynomial running time.
Now we turn to the setting $\cdriv = \cwait$ with parking locations and show that it is still
tractable.

Cost profiles are piecewise linear functions.
An important aspect of our polynomial time proof is to count the non-differentiable points of the profiles.
The running time of each profile operation of our algorithm is linear in the number of non-differentiable points of the involved profiles.
These points are either \emph{convex}, \emph{concave}, or \emph{discontinuous}, meaning an environment around such a point exists in which the profile is convex or concave or discontinuous, respectively.
In a discontinuous point, a profile is always jumping down.

The non-differentiable points in the cost profiles are induced by the travel time functions.
In our example of \cref{fig:exampleTravelTimeFunction}, the convex points are $\{4,8,11\}$, the concave points are $\{6,9,12\}$, and the discontinuous points are $\{10,13,15\}$.
For a travel time function $\T_\edge$ of an edge $\edge$, we can assign a convex point $t$ to the beginning of a ban interval in $t$, a concave point $t$ to the end of a ban interval in $t$, and a discontinuous point $t$ to the end of a ban interval in $t-\T_\edge(t)$.
From this initial assignment, we can derive a ban interval assignment of the convex or discontinuous points of cost profiles.
We omit to count the number of concave points of a cost profile because every gradient of a piece must be in $\{\cwait,\ldots,\cwait[\maxCat]\}$, so the number of consecutive concave points in a cost profile is limited by $\maxCat$.

Initially, a profile $\C_{v}$ of a vertex $v$ has no convex or discontinuous points.
Such points may be introduced in the third step of an iteration of the algorithm when the auxiliary profile $\C'_{v}$ is merged into $\C_{v}$.
In the second step of an iteration, no new convex or discontinuous points can arise in $\C'_{v}$, so all such points must be created in $\C'_{v}$ in the first step.
Since $\cdriv=\cwait$, $\C'_v(t)$ is set to $\C_u(t - \T_e(t)) + \cdriv \cdot \T_e(t)$ (compare \cref{eqn:step1}) for some edge $e=(u,v)$ in this step.
If $t_v$ is a convex or discontinuous point of $\C'_v$, then $\T_e$ must be convex or discontinuous in the same point in time, or $\C_u$ must be convex or discontinuous in $t_u:=t_v - \T_e(t_v)$.
In the former case, $t_v$ inherits the assignment of the same point in time in $\T_e$, whereas in the latter case, $t_v$ inherits the assignment of $t_u$ in $\C_u$.
Since the cost profiles change during the algorithm, we do not only assign a ban interval to every convex or discontinuous point but also an iteration.
Again, in the former case, $t_v$ is assigned the current iteration, whereas in the latter case, $t_v$ inherits the iteration assignment of $t_u$ in $\C_u$.

\begin{lemma}
\label{lemma:numberOfPointsInCostProfile}
If $\cdriv = \cwait$ then a cost profile after iteration $i$ has at most $ib$ convex and at most $ib$ discontinuous points.
\end{lemma}

\begin{proof}
In the following, we denote the state of the profile $C_v$ after iteration $i$ by $C^i_v$.
Let $t_v$ be a convex or discontinuous point of $\C^i_{v}$ that is assigned both to an iteration $k$ and to a ban interval of some edge with head $x$.
We can follow the inheritance relation until we finally reach a convex or discontinuous point $t_x$ in $\C^k_x$.
By induction, we have $\C^i_v(t_v) = \C^k_x(t_x) + \cdriv \cdot (t_v - t_x)$.
Now suppose there are two convex or two discontinuous points $t^1_v < t^2_v$ in the profile $\C^i_{v}$ that are assigned to the same ban interval and the same iteration $k$, so they can be traced back to the same point $t_x$ in $\C^k_x$.
Then the previous observation implies that $\C^i_v(t^2_v) - \C^i_v(t^1_v) = \cdriv \cdot (t^2_v - t^1_v)$ holds, that is, the profile $\C^i_{v}$ must contain a piece with gradient $d$ that contains both $t^1_v$ and $t^2_v$.
But then $t^2_v$ can neither be convex nor discontinuous.
Hence, two convex or two discontinuous points must differ in their assigned ban interval or their assigned iteration and there can only be $ib$ discontinuous and convex points, respectively.
\end{proof}

\begin{restatable}{lemma}{iterations}
\label{lemma:iterations}
If $\cdriv = \cwait[0]$ then the total number of iterations is at most $2n(b(r+1)+1)$.
\end{restatable}

As in the proof of Lemma~\ref{lemma:numberOfPointsInCostProfile} we use
the ban interval assignment of convex and discontinuous points. Every visit of a
vertex can either be assigned to the start or end of a ban interval, or it can be
assigned to a concave point of the final cost profile of the vertex. The detailed
proof is omitted due to space limitations and can be found in
Appendix~\ref{sec:poly_iterations}.

\begin{theorem}
If $\cdriv = \cwait[0]$ then the running time of the algorithm is polynomial.
\end{theorem}

\begin{proof}
From Lemma~\ref{lemma:numberOfPointsInCostProfile} with the bound from
Lemma~\ref{lemma:iterations} it follows that the number of pieces of any profile
that is constructed during the algorithm is polynomial.

We now estimate the overall running time of our algorithm:
Lemma~\ref{lemma:iterations} states that the total number of iterations is
polynomial.
In every iteration of the algorithm one vertex is considered and for its
outgoing edges the profiles are updated with a running time linear in
the number of pieces of the profiles.
The adjacent vertices are inserted into the priority queue or their keys are
decreased. Since the size of the priority queue is at most the total number
of vertices also the running time of the priority queue operations is
polynomial.
\end{proof}

\section{Implementation}
\label{sec:impl}

The past decade has seen a lot of research effort on the engineering of efficient route planning algorithms.
This section describes the speed-up techniques we employ in our implementation and some implementation details.

We store cost profiles as a sorted list of pieces.
Each piece is represented as a triple: a point in time from which this piece is valid, the costs it takes to reach the vertex at the beginning of the piece and the incline of the piece.
For each piece we also store a parent vertex.
This allows us to efficiently reconstruct routes by traversing the parent pointers.

We employ A* to guide the search toward the destination vertex.
The queue is ordered by the original key plus an estimate of the remaining distance (here: driving time) to the destination.
The estimate for vertex $\vertexa$ is denoted by $\pi_\target(\vertexa)$.
We use the exact shortest driving time to $\target$ without driving restrictions as the potential.
This is the best possible potential in our case.
We efficiently extract these exact distances from a Contraction Hierarchy \cite{gssv-erlrn-12}, as described in \cite{strasser2019perfect}.
Since our algorithm has to run until the queue is empty, we can not immediately terminate when we reach the destination.
However, we get a tentative cost profile at the destination.
This allows for effective pruning.
Additionally, we do not need to insert a vertex $\vertexa$ into the queue when $\update + \pi_\target(\vertexa) > \Hmax$ holds, that is, we cannot reach the destination from $\vertexa$ within the planning horizon.

We employ pruning to avoid linking and merging when possible using the following rules:
\begin{itemize}
	\item Consider a vertex $\vertexa$  that is visited at $\update$.
		Before relaxing any outgoing edges, we first check if $\vertexa$ can actually contribute to any optimal route to $\target$.
		If $\C_\vertexa(\atime) + \pi_\target(\vertexa) \cdot \cdriv > \C_\target(\atime + \pi_\target(\vertexa))$ for all $\atime$ with $\update \leq \atime < \Hmax$, $\vertexa$ can not contribute to an optimal route to $\target$ and can thus be skipped.
	\item Let $\alpha(\vertexa) := \min\{\atime \mid \C_\vertexa(\atime) < \infty\}$ be the first point in time such that $\vertexa$ can be reached with finite costs and $\infty$ if no such point exists.
    For each vertex $\vertexa$, we maintain a lower bound $\beta(\vertexa) := \min_\atime\{\C_\vertexa(\atime)\}$ and an upper bound $\gamma(\vertexa) := \max_{\atime > \alpha(\vertexa)}\{\C_\vertexa(\atime)\}$ or $\infty$, if there are no finite costs.
    They can be updated efficiently during the merge operation.
    An edge $(\vertexa, \vertexb)$ only needs to be relaxed if $\beta(\vertexa) + \drive(\vertexa, \vertexb) \cdot \cdriv \leq \gamma(\vertexb)$ or $\alpha(\vertexa) + \drive(\vertexa, \vertexb) < \alpha(\vertexb)$.
	\item When all of the pieces of the cost profile of a vertex $\vertexa$ share the same parent vertex $\vertexb$ and $\categ(\vertexa) = 0$, the edge $(\vertexa, \vertexb)$ back to the parent does not need to be relaxed as loops can never be part of an optimal route unless they include waiting at a parking location.
\end{itemize}

\section{Experimental Evaluation}
\label{sec:exp}

Our algorithm is implemented in C++14 and compiled with Visual C++.
For the CH-potentials, we build upon the Contraction Hierarchy implementation of RoutingKit\footnote{\url{https://github.com/RoutingKit/RoutingKit}}~\cite{dsw-cch-15}.
All experiments were conducted on a Windows 10 Pro machine with an Intel i7-7600 CPU with a base frequency of 3.4\,GHz and 32\,GB of DDR4 RAM.
The implementation is single-threaded.

Our experimental setup is taken from~\cite{b-rptrc-18}.
We perform experiments on a road network used in production by PTV\footnote{\url{https://ptvgroup.com}}.
The network is adapted from data by TomTom\footnote{\url{https://tomtom.com}}.
It covers Austria, France, Germany, Italy, Liechtenstein, Luxembourg, and Switzerland.
It has 21.9 million vertices and 47.6 million edges.
We use travel times, driving bans, and road closures for a truck with a gross combined weight of 40~tons.
Driving bans were derived from the current legislation of the respective countries.
This includes Sunday driving bans in all countries, a late Saturday driving ban in Austria and night driving bans in Austria, Liechtenstein and Switzerland.
Additionally, there is a Saturday driving ban in Italy during the summer holidays.
The dataset also includes several local road closures in city centers.

Parking locations were taken from data by Truck Parking Europe\footnote{\url{https://truckparkingeurope.com}}.
There is a total of 15\,317 vertices classified as parking locations in our data set.
The dataset also contains the capacity of each parking location.
We assign each parking location a rating between 1 and 5 depending on its capacity.
Table~\ref{tab:parking_lots} shows the number of parking locations for each rating and our default waiting costs.
We also evaluate different parameterizations. 
The waiting costs are calculated such that for an hour of waiting a detour of up to four minutes will be taken to get to a parking location rated better by one.
For waiting at the source vertex of a query, we assign zero waiting costs regardless of the rating.

\begin{table}
\centering
\caption{Rating and default waiting cost by capacity of parking locations. The driving cost is the same as the cost for waiting at unrated vertices.}\label{tab:parking_lots}
\begin{tabular}{lrrrrrr}
\toprule
Capacity of parking locations & $\geq 80$ & $\geq 40$ & $\geq 15$ & $\geq 5$ &  $\geq 1$ &       - \\
\midrule
Rating                        &         5 &      4 &      3 &      2 &      1 &       0 \\
Default waiting costs         &         3 &      4 &      5 &      6 &      7 &      14 \\
Number of parking locations   &       448 &    997 & 2\,664 & 5\,418 & 5\,748 & 21.9\,M \\
\bottomrule
\end{tabular}
\end{table}

We generate two sets of source-destination pairs and combine them with different planning horizons.
The first set is used to evaluate the practicality of our model.
It is designed to make the algorithm cope with the night driving ban in Austria and Switzerland.
We select 100 pairs of vertices.
One vertex is randomly selected from the area around southern Germany.
The other vertex is selected from the area around northern Italy.
See Figure~\ref{fig:geofence_vis} for exact coordinates and a visualization.
We store each pair in both directions.
Hence, we have 200 vertex pairs in this set.
The planning horizon starts at Monday 2018/7/2, 18:00 with length one day (query set A1) and two days (A2).
Figure~\ref{fig:result_example} depicts an example query from A1.

The second set is generated by selecting 100 source vertices uniformly at random.
From each source vertex, we run Dijkstra's algorithm without a specific target ignoring any driving restrictions.
Dijkstra's algorithm explores the graph by traversing vertices in increasing distance of the source vertex.
We use the order in which vertices are settled to select destination vertices with different distances from the source.
Every $2^i$th settled vertex with $i \in [12,24]$ is stored.
We denote $i$ as the \emph{rank} of the query.
This results in 1\,300 source-destination pairs.
We combine these vertex pairs with four planning horizons: starting at Friday 2018/7/6, 06:00 for one day (denoted as query set B1), for two days (B2) and starting later that day at 18:00 for one day (B3) and for two days (B4).

\begin{figure}
\centering
\includegraphics[width=.5\textwidth]{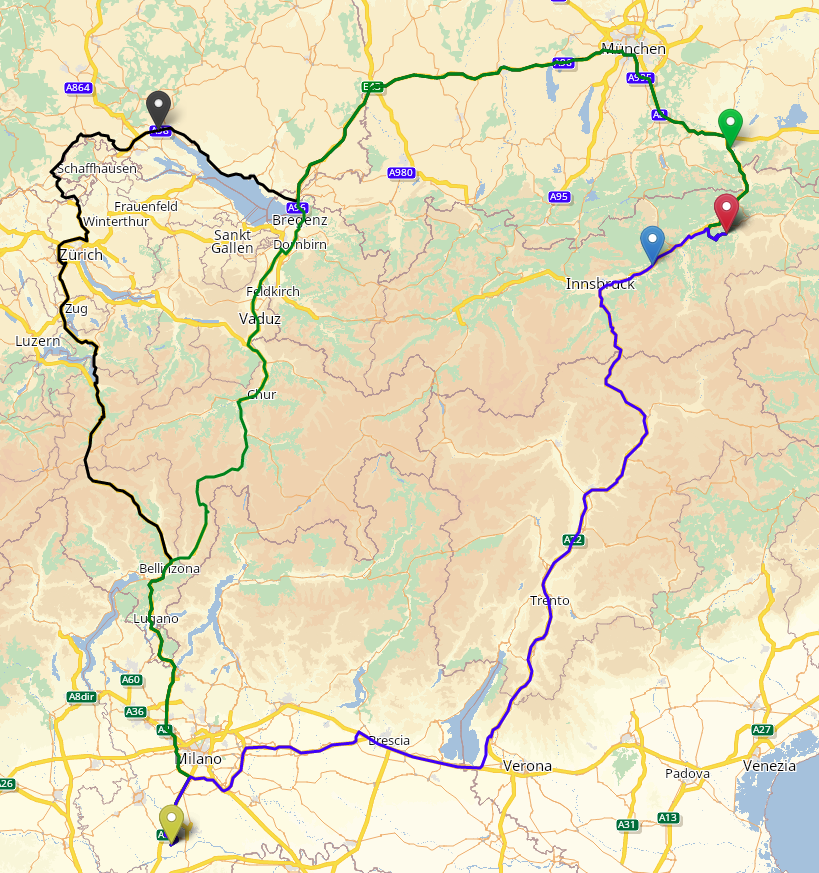}
\caption{
Optimal paths of an example query from northwestern Austria to northern Italy, slightly south of Milano.
The source is indicated by a red, the destination by a yellow marker.
The other markers indicate the parking locations along the respective routes.
The blue route in the east has the shortest driving time, around 10.5 hours, but the latest arrival.
It schedules a waiting time of seven hours during the night driving ban at a parking location of rating 4 and afterwards takes the fastest route to the destination.
The green route in the middle arrives an hour earlier at the destination but the driving time is over two hours longer.
This route includes three hours of waiting at a parking location of rating 5.
The black route in the west takes 16 hours to drive, includes only a few minutes of waiting and arrives six minutes before the green one.
}
\label{fig:result_example}
\end{figure}

\begin{table}
\centering
\caption{
Query statistics for different waiting cost parameters for query set A1.
The first six columns show the waiting cost parameters.
Waiting costs at the source are always set to zero.
The waiting time columns depict the share of the time spent waiting at vertices with the respective rating summed up over all routes.
The routes column gives the average number of optimal routes per query.
The arrival time deviation column contains the average of the difference between earliest and latest arrival time among all optimal routes for all queries.
Running times are also averaged.
}\label{tab:cost_params}
\begin{tabular}{r@{\hskip4pt}r@{\hskip4pt}r@{\hskip4pt}r@{\hskip3pt}r@{\hskip2pt}rr@{\hskip5pt}r@{\hskip5pt}r@{\hskip5pt}r@{\hskip5pt}r@{\hskip5pt}r@{\hskip5pt}r@{\hskip5pt}r@{\hskip5pt}r@{\hskip5pt}r@{\hskip5pt}r}
\toprule
{} & {} & {} & {} & {} & {} & {} & {} & {} & {} & {} & {} & {} & Optimal & Arrival time & Running \\
{} & {} & {} & {} & {} & {} & \multicolumn{7}{c}{Waiting time by rating [\%]} & Routes & deviation & time \\ \cmidrule(lr){7-13}
$\cwait[5]$ & $\cwait[4]$ & $\cwait[3]$ & $\cwait[2]$ & $\cwait[1]$ & $\cwait[0] = \cdriv$ & $\source$ & 5 & 4 & 3 & 2 & 1 & 0 & [\#] & [h:mm] & [ms] \\
\midrule
1  & 10 & 50 & 100 & 1000 & 10000 & 59.4 & 2.5 & 5.8 & 23.3 & 3.1 & 2.8 & 3.1 & 3.02 & 2:21 & 364.1 \\
1  & 2  & 4  & 8   & 16   & 128   & 62.2 & 3.5 & 6.5 & 19.8 & 2.1 & 2.8 & 3.1 & 3.02 & 2:20 & 412.3 \\
1  & 2  & 4  & 8   & 16   & 32    & 70.8 & 6.0 & 5.1 & 12.1 & 1.0 & 1.9 & 3.1 & 2.96 & 2:20 & 435.4 \\
3  & 4  & 5  & 6   & 7    & 14    & 79.3 & 6.2 & 3.1 &  4.6 & 1.5 & 2.0 & 3.3 & 2.86 & 2:17 & 529.4 \\
16 & 24 & 28 & 30  & 31   & 32    & 85.2 & 4.6 & 1.1 &  3.3 & 1.1 & 1.1 & 3.6 & 2.71 & 2:14 & 742.2 \\
\bottomrule
\end{tabular}
\end{table}

We first investigate whether allowing waiting everywhere (albeit penalized) may lead to unwanted results in practice.
On the one hand, routes with many stops are impractical.
Our experiments indicate that this is not the case:
Accross all routes for A1, there is at most one additional stop scheduled (0.2 on average).
On the other hand, let us call a route \emph{precarious} if waiting is scheduled at an unrated location (other than the source vertex).
For 187 of the 200 queries of A1, there is no precarious route in the Pareto set.
For the other 13 queries, the Pareto set always contains more than one route, and it is always only the quickest route in the Pareto set that is precarious.
So filtering out such routes in a postprocessing step does not make a query infeasible.
On average, the second quickest route in the Pareto set arrives 422\,s later than the quickest but precarious route (minimum 38\,s, maximum 877\,s).

We also evaluate the influence of different waiting cost parameterizations on the performance and the results of our algorithm.
Table~\ref{tab:cost_params} depicts the results.
We observe that the parametrization has only limited influence on the results of the algorithm.
The average number of optimal routes and the arrival time deviation change only very little even between the two most extreme configurations.
Since waiting at the source vertex costs nothing, the majority of the waiting in all configurations is scheduled there.
When waiting at parking locations is much cheaper than driving, less waiting time will be scheduled at the source and more waiting at parking locations.
Also, clear differences between the costs lead to a better running time, because cost profiles become less complex.

\begin{table}
\centering
\caption{Query statistics for all six query sets.
First, for all queries. Second, only for non-trivial queries.
A query is denoted as trivial if there is exactly one optimal route which is also optimal when ignoring all driving restrictions.
All numbers are averages unless reported otherwise.
The arrival time deviation column contains the average of the difference between earliest and latest arrival time among all optimal routes for all queries.
The routes column contains the number of optimal routes.}\label{tab:perf}
\begin{tabular}{r@{\hskip4pt}r@{\hskip4pt}lrrrrrr}
\toprule
&     &                   & Query & Optimal & Arrival time & \multicolumn{2}{c}{Running time} \\ \cmidrule(lr){7-8}
&     &                   & share &  Routes &    deviation &      Avg. &               Median \\
& Set & Planning horizon  &  [\%] &    [\#] &       [h:mm] &      [ms] &                 [ms] \\
\midrule
& A1 & Mon. 18:00, 1 day  & 100.0 &    2.86 &         2:17 &     529.4 &                266.3 \\
& A2 & Mon. 18:00, 2 days & 100.0 &    3.54 &         3:19 &     648.1 &                405.6 \\
& B1 & Fri. 06:00, 1 day  & 100.0 &    1.04 &         0:10 &      10.0 &                  0.6 \\
& B2 & Fri. 06:00, 2 days & 100.0 &    1.08 &         0:16 &      79.5 &                  0.7 \\
& B3 & Fri. 18:00, 1 day  & 100.0 &    1.13 &         0:08 &     205.8 &                  0.6 \\
& B4 & Fri. 18:00, 2 days & 100.0 &    1.32 &         0:20 &  1\,028.1 &                  0.7 \\
\midrule
\parbox[t]{3mm}{\multirow{6}{*}{\rotatebox[origin=c]{90}{Only non-trivial}}}
& A1 & Mon. 18:00, 1 day  &  67.5 &    3.82 &         3:13 &     764.1 &                560.6 \\
& A2 & Mon. 18:00, 2 days &  72.0 &    4.53 &         4:37 &     899.2 &                655.0 \\
& B1 & Fri. 06:00, 1 day  &   4.1 &    2.19 &         4:10 &      42.5 &                  6.6 \\
& B2 & Fri. 06:00, 2 days &   4.8 &    2.76 &         5:43 &  1\,105.6 &                 35.8 \\
& B3 & Fri. 18:00, 1 day  &   9.2 &    2.73 &         1:25 &  1\,359.0 &                475.2 \\
& B4 & Fri. 18:00, 2 days &  11.6 &    3.79 &         2:51 &  5\,819.4 &             1\,947.2 \\
\bottomrule
\end{tabular}
\end{table}

We next investigate the algorithm's performance for each of the different query sets.
We report the same numbers limited to non-trivial queries.
A query is denoted as \emph{trivial} if there is exactly one optimal route which is also optimal when ignoring all driving restrictions.
Table~\ref{tab:perf} depicts the results.
Clearly, the query set has a strong influence on the running time of the algorithm.
Average running times range from ten milliseconds to one second when looking at all queries.
However, median query times are significantly smaller.
The reason for this is that our algorithm can answer trivial queries in a few milliseconds or less.
Due to the perfect potentials, the algorithm only traverses the optimal path.
Once the destination is reached, because of the target pruning, all other vertices in the queue are skipped and the algorithm terminates.
Excluding trivial queries, we get a clearer picture of the algorithm's performance when solving the harder part of the problem.

For the query sets B1 and B2, only 4\% to 5\% of the queries have to deal with driving restrictions.
This is mostly due to closures for individual roads in certain cities and not country-wide driving bans.
When the planning horizon begins later at 18:00 (B3 and B4), we get around twice as many non-trivial queries.
These are primarily caused by the night driving bans in Austria and Switzerland.
Road closures and country-wide driving bans lead to different optimal routes.
When there is a road closure on the shortest path ignoring any driving restrictions, we often have two optimal routes.
One which takes a (small) detour around the closure, and one waiting at the source until the closed road opens and then taking that slightly shorter path.
Thus, we have two routes with very similar driving times but (often vastly) diverging arrival times.
When dealing with night driving bans, we get more optimal results with different trade-offs as in the example of Figure~\ref{fig:result_example}.

Increasing the length of the planning horizon to two days leads to more non-trivial queries, more optimal routes per query, and a greater deviation in arrival time.
The reason are routes with a travel time longer than 24 hours which were not valid for the shorter planning horizon.

Even when we restrict ourselves to queries with non-trivial results, running times still vary depending on the query set.
Average and median deviate not as strong as when considering all queries, but the distribution of running times is still skewed by a few long running queries, especially on set B4.
The reason for this is that the running time heavily depends on the types and lengths of driving restrictions in the search space.
The Saturday driving ban in Italy causes heavy outliers in B4 (but also B2 and B3), when the destination lies in an area blocked for most of the planning horizon.
This causes the algorithm to explore large parts of the graph, until the driving ban is over.
The worst of these queries took 49 seconds to answer.
Nevertheless, when looking at query sets A1 and A2, we clearly see that the algorithm can answer queries affected by country-wide night driving bans in less than a second.

\section{Conclusion}
\label{sec:conclusion}

We have introduced a variant of the shortest path problem where driving on edges may be forbidden at times, both driving and waiting entail costs, and the cost for waiting depends on the rating of the respective location.
The objective is to find a Pareto set of both quickest paths and minimum cost paths in a road graph.
We have presented an exact algorithm for this problem and shown that it runs in polynomial time if the cost for driving is the same as for waiting in an unrated location.
With this algorithm, we can solve routing problems that arise in practice in the context of temporary driving bans for trucks as well as temporary closures of roads or even larger parts of the road network.

Our experiments demonstrate that our implementation can answer queries with realistic driving restrictions in less than a second on average.
There are a few slow outlier queries when the destination vertex lies in a blocked area.
A promising angle to improve this could be to study bidirectional variants of our algorithm.
We exploit Contraction Hierarchies to efficiently obtain good A* potentials.
The algorithm can also be used in a dynamic (or live or online) scenario when combined with Customizable Contraction Hierarchies~\cite{dsw-cch-15}.
A natural extension of our problem at hand is to consider time-dependent driving times or rules for truck drivers that enforce a break after a certain accumulated driving time.

\appendix

\section{Driving Cost Less Than Unrated Waiting Cost}
\label{sec:appendixA}

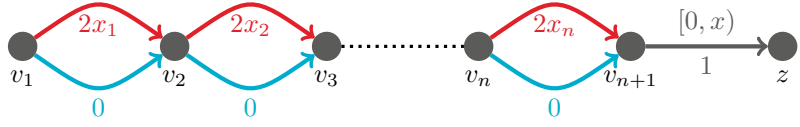
\begin{figure}
    \begin{center}
        \begin{tikzpicture}

    \node[circle,draw,fill,grau](v1) at (0,0){};
    \node[circle,draw,fill,grau](v2) at (2,0){};
    \node[circle,draw,fill,grau](v3) at (4,0){};
    \node[circle,draw,fill,grau](vn) at (6,0){};
    \node[circle,draw,fill,grau](vnplus1) at (8,0){};
    \node[circle,draw,fill,grau](z) at (10,0){};

    \node[below] at (v1.south) {$v_1$};
    \node[below] at (v2.south) {$v_2$};
    \node[below] at (v3.south) {$v_3$};
    \node[below] at (vn.south) {$v_n$};
    \node[below] at (vnplus1.south) {$v_{n+1}$};
    \node[below] at (z.south) {$z$};
    
    \draw[->,rot]  (v1) .. controls +(1,0.7) .. (v2) node [below,midway] {$2x_1$} ;
    \draw[->,blau] (v1) .. controls +(1,-0.7) .. (v2) node [below,midway] {$0$} ;

    \draw[->,rot]  (v2) .. controls +(1,0.7) .. (v3) node [below,midway] {$2x_2$} ;
    \draw[->,blau] (v2) .. controls +(1,-0.7) .. (v3) node [below,midway] {$0$} ;

    \draw[dotted, very thick] (v3) -- (vn);
 
    \draw[->,rot]  (vn) .. controls +(1,0.7) .. (vnplus1) node [below,midway] {$2x_n$} ;
    \draw[->,blau] (vn) .. controls +(1,-0.7) .. (vnplus1) node [below,midway] {$0$} ;
   
    \draw[->,grau]  (vnplus1) -- (z) 
        node [below,midway] {$1$} 
        node [above,midway] {$[0,x)$} ;

\end{tikzpicture}
    \end{center}
	\caption{Transformation of a PARTITION instance consisting of $n$ numbers $x_i$ into a road
            graph with temporary driving bans. The last edge is closed before $x :=  \sum_{i=1}^n x_i $.
            All vertices have rating 0. The graph has parallel edges and edges with driving time $0$
            for the sake of convenience. This can be avoided by replacing each lower blue edge by
            two edges with driving time $1$, adding $2$ to the driving time of each upper red edge,
            and adding $2n$ to $x$.}
    \label{fig:graphPartition}
\end{figure}

\nphard*

\begin{proof}
It can be verified in polynomial time if a route is feasible and has travel time costs less than or equal to $k$,
so the problem is in $\NP$. To show the $\NP$-completeness we reduce from PARTITION:
Given a set of $n$ numbers $x_i$, we construct in polynomial time a road graph with time windows
and costs as shown in Figure~\ref{fig:graphPartition}. The only time window is on the last edge
which is closed up to time $x :=  \sum_{i=1}^n x_i$. Let $\cdriv < \cwait$ and set the threshold
for the travel time cost to $k:=\cdriv(x + 1)$.

If there is a partition of the $x_i$ into two subsets $S_1$ and $S_2$ with the same sum $x/2$ then there
is a route with travel time costs $k$: For $(v_i, v_{i+1})$ select the upper red edge with driving time $2x_i$
if $x_i \in S_1$, and the lower blue edge with driving time $0$ if $x_i \in S_2$. Without waiting this
route arrives exactly at $v_{n+1}$ at time $x$ and hence arrives at time $x+1$ at the destination $z$ with
travel time cost $\cdriv(x + 1) = k$.

On the other hand, if there is a route with travel time cost $c \le k$, there is
a valid partition of the $x_i$: Since the last edge is traversable not earlier than
time $x$ and $\cdriv < \cwait$, any waiting time on the route implies $c > \cdriv (x+1) = k$,
so there can be no waiting time included in the route, $c = \cdriv (x+1)$, and the driving time of
the route is $x+1$. Set $S_1$ to the $x_i$ of all upper red edges of the route and $S_2$ to the
remaining $x_i$. The last edge has driving time 1, so the driving time from $v_1$ to $v_{n+1}$
equals $x$ and consists solely of the driving times of the upper red edges in the route. We
conclude that $\sum_{S_1} 2 x_i = x$, which implies
$\sum_{S_1} x_i = (\sum_{i=1}^n x_i)/2 = \sum_{S_2} x_i$.
\end{proof}

\begin{figure}
    \begin{center}
        \begin{tikzpicture}

    \node[circle,draw,fill,grau](v1) at (0,0){};
    \node[circle,draw,fill,grau](v2) at (2,0){};
    \node[circle,draw,fill,grau](v3) at (4,0){};
    \node[circle,draw,fill,grau](v4) at (6,0){};
    \node[circle,draw,fill,grau](vi) at (8,0){};
    \node[circle,draw,fill,grau](vi1) at (10,0){};
    \node[circle,draw,fill,grau](vnplus1) at (12,0){};

    \node[below] at (v1.south) {$v_1$};
    \node[below] at (v2.south) {$v_2$};
    \node[below] at (v3.south) {$v_3$};
    \node[below] at (v4.south) {$v_4$};
    \node[below] at (vi.south) {$v_i$};
    \node[below] at (vi1.south) {$v_{i+1}$};
    \node[below] at (vnplus1.south) {$v_{n+1}$};
    
    \draw[->,rot]  (v1) .. controls +(1,0.7) .. (v2) 
        node [below,midway] {$2$} 
        node [above,midway]{$[1, 1 + x)$} ;
    \draw[->,blau] (v1) .. controls +(1,-0.7) .. (v2) node [below,midway] {$x$} ;

    \draw[->,rot]  (v2) .. controls +(1,0.7) .. (v3) 
        node [below,midway] {$4$} 
        node [above,midway,text width=2cm,align=center]{ $[3 + x$, \\ $3 + 3x)$ } ;
    \draw[->,blau] (v2) .. controls +(1,-0.7) .. (v3) node [below,midway] {$2x$} ;

    \draw[->,rot]  (v3) .. controls +(1,0.7) .. (v4) 
        node [below,midway] {$8$} 
        node [above,midway,text width=2cm,align=center]{$[7 + 3x$, $7 + 7x)$} ;
    \draw[->,blau] (v3) .. controls +(1,-0.7) .. (v4) node [below,midway] {$4x$} ;

    \draw[dotted, very thick] (v4) -- (vi);
    
    \draw[->,rot]  (vi) .. controls +(1,0.7) .. (vi1) 
        node [below,midway] {$2^{i}$} 
        node [above,midway]{$[b_i, b_i + 2^{i-1}x)$} ;
    \draw[->,blau] (vi) .. controls +(1,-0.7) .. (vi1) node [below,midway] {$2^{i-1}x$} ;

    \draw[dotted, very thick] (vi1) -- (vnplus1);

\end{tikzpicture}
    \end{center}
	\caption{A graph with exponentially many Pareto-optimal routes if $\cdriv > \cwait[0]$. The ban
	    interval of an upper red edge $(v_i, v_{i+1})$ begins at $b_i:=(2^i-1)+(2^{i-1}-1) x$.
	    There is no parking location in this graph, so all vertices have rating 0. The graph has
	    parallel edges for the sake of convenience, they can be avoided by replacing
	    each lower blue edge by two edges and splitting the driving time.}
    \label{fig:graphExponentialSolutions}
\end{figure}
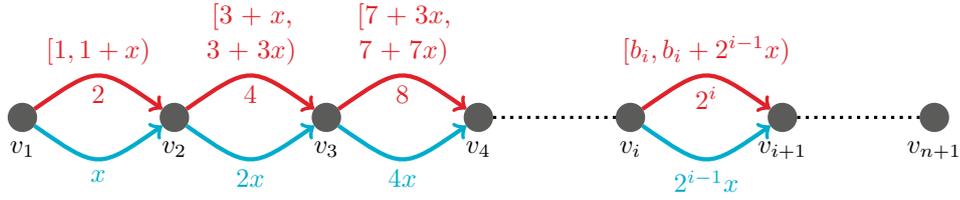

\section{Driving Cost Greater Than Unrated Waiting Cost}
\label{sec:appendixB}

\exproutes*

\begin{proof}
Given $\cdriv > \cwait[0]$, let $x := \left\lceil \frac{2 \cdriv}{\cdriv-\cwait} \right\rceil+1$ be a
possibly large constant. Consider the graph shown in Figure~\ref{fig:graphExponentialSolutions} with
$n+1$ vertices and $2n$ edges.
The graph has no parking locations, so waiting costs are independent of the location. Given a
path, a route with that path which arrives as soon as possible has also minimal travel
time cost.

In order to calculate the earliest possible arrival time of a route with
a given path we assign every $v_1$-$v_{k+1}$ path a number $p \in \NNI$ with $ 0 \le p < 2^k$:
In the binary representation of $p$ beginning with the least significant digit the $i$-th digit is 1 if and only
if in the path the vertices $v_i$ and $v_{i+1}$ are connected by the upper red edge. Let $\overline{p}$ be the
ones' complement of $p$, so $p + \overline{p} = 2^k-1$.
With this representation of $p$ and assuming that for upper red edges always a waiting time for the whole
duration of the ban interval is required, the earliest possible arrival time at $v_{k+1}$ is
\begin{equation}
    \begin{split}\label{eqn:arrivalOfPath}
        \textrm{arrival}(p) = \overline{p} x + p (2 + x) = (\overline{p} + p) x + 2 p =  (2^k-1) x + 2 p \\
    \end{split}
\end{equation}

By induction we show that whenever an upper red edge is used in the path, waiting the whole
ban interval is required: For $k=1$ the path consists of one edge and
the departure from $v_1$ is at time 0. If $p=0$ the edge is the lower blue edge
with driving time $x$ and the arrival at $v_2$ is at time $x$. If $p=1$ the edge is the
upper red edge with driving time 2, and the ban interval starts at time 1 during
traversal of the edge. A waiting time of $x$ is required, and the arrival at $v_2$ is at
time $2 + x$.
Assume now that the proposition holds for $v_1$-$v_{k+1}$-paths with $k$ edges and consider a
$v_1$-$v_{k+2}$-path with $k+1$ edges. If the last edge is the lower blue edge there is
nothing to show, so further assume the last edge connecting $v_{k+1}$ and $v_{k+2}$ is
the upper red edge. According to Equation~\ref{eqn:arrivalOfPath} the arrival at
$v_{k+1}$ is between $(2^k-1) x$ and $(2^k-1) x + 2^{k+1}-2$.
The driving time for $(v_{k+1},v_{k+2})$ is $2^{k+1}$ and the ban interval begins at
$b_{k+1}=(2^{k+1}-1)+(2^k-1) x$. Hence, in all cases during the traversal of the edge the
ban interval begins.

The travel time cost of the route can be calculated as follows.
A lower blue edge $(v_i, v_{i+1})$ has only driving cost of $2^i \cdriv x$.
An upper red edge $(v_i, v_{i+1})$ has driving cost $2^{i+1} \cdriv$ and
waiting cost $2^i \cwait x$:

\begin{equation}
    \begin{split}\label{eqn:costOfPath}
        \textrm{cost}(p) & = \overline{p} \cdriv x + 2 p \cdriv + p \cwait x \\
	                 & = (2^k - 1 - p) \cdriv x  + 2 p \cdriv + p \cwait  x \\
                         & = (2^k - 1) \cdriv x - p \cdriv x + 2 p \cdriv + p \cwait  x \\
			 & = (2^k - 1) \cdriv x + p (2 \cdriv - x (\cdriv - \cwait))
    \end{split}
\end{equation}

Let $p_1$ and $p_2$ with $p_1 < p_2$ be the assigned numbers of two arbitrary $v_1$-$v_{n+1}$-paths
and consider corresponding routes with the earliest possible arrival time and travel time costs
as calculated above.
From Equation~\ref{eqn:arrivalOfPath} it follows directly that
$\textrm{arrival}(p_1) < \textrm{arrival}(p_2)$.  By definition of $x$, we have
$2 \cdriv - x (\cdriv - \cwait) < 0$ and together with
Equation~\ref{eqn:costOfPath} it follows that $\textrm{cost}(p_1) > \textrm{cost}(p_2)$.
This means that each of the $2^n$ routes is Pareto-optimal because compared to any other route
either the arrival time is earlier or the travel time cost is lower.
\end{proof}

\section{Polynomial Number of Iterations}
\label{sec:poly_iterations}

\iterations*

\begin{proof}
Similarly to the proof of Lemma~\ref{lemma:numberOfPointsInCostProfile}, we show that every vertex is visited at most $\mathcal{O}(rb)$ times.
For this, we uniquely assign the elements of the priority queue to the start or end of a ban interval.
When a vertex is inserted into the priority queue the key $t^*$ is the earliest point in time such that $\C'_v(t^*)<\C_v(t^*)$.
Such a $t^*$ always coincides with a non-differentiable point of a cost profile, though this point may not necessarily be included in the final cost profile.
As the number of concave points between two convex or discontinuous points is at most $r$, we only track priority queue elements induced by convex or discontinuous points.
Additionally, there is one first point of each profile without an assignment to a ban interval.

To assign the points, we distinguish several cases depending on the time $t^*-1$:
\begin{itemize}
    \item If $\C_v$ is $\infty$ at $t^*-1$ then $(t^*, \C'_v(t^*))$ is the new first point of the resulting profile at $v$.
          Due to the correctness of the algorithm, the queue element for the previous first point is still in the queue.
          It will be updated and its key decreased.
    \item Otherwise, if $\C'_v(t^*-1) > C'_v(t^*)$ and there is a jump discontinuity in $C'_v(t^*)$ then there is a corresponding ban interval.
          We assign $t^*$ to the end of that corresponding ban interval.
          This also covers that case that $\C'_v(t^*-1) = \infty$.
          In this case a corresponding ban interval has to exist, too.
          For contradiction, assume that no such ban interval exists.
          Thus, waiting is never beneficial and the entire time was spent driving and $\C'_v(t^*) = \cdriv(t^* - \Hmin)$.
          However, by definition $\C'_v(t^*)<\C_v(t^*)$ which is a contradiction since $\cdriv(t^* - \Hmin)$ is the maximum possible cost value at $t^*$.
          Thus, a ban interval has to exist and can be used for the assignment.
    \item Otherwise, if $\C'_v(t^*-1) > \C_v(t^*-1)$ then there is a concave point between $t^*-1$ and $t^*$.
    \item Otherwise, if $\C'_v(t^*-1) = \C_v(t^*-1)$ then $t^*-1$ is either a concave point of $\C'_v$ or a convex point of $\C_v$.
          In the case of a convex point, we assign $t^*$ to the start of the corresponding ban interval.
          Note that the first piece of the resulting profile from $t^*-1$ to $t^*$ has an incline less than $d$.
\end{itemize}
We show by contradiction that a vertex at time $\update$ which is assigned to a specific start or end of a ban interval is visited only once:
Assume a vertex is visited again at $t_2$ with the same assignment as a previous visit at $t_1$.
Due to the correctness of the algorithm all cost profiles up to $t_2$ are correct.
If the points are assigned to the same start of a ban interval, we have $C_v(t_2-1) = C_v(t_1-1) + d(t_2-t_1)$.
This is a contradiction to the remark above that the incline at $t_1-1$ is less than $d$.
If the points are assigned to the same end of a ban interval, we have $C_v(t_2) = C_v(t_1) + d(t_2-t_1)$.
If the incline at $t_1$ was less than $d$, this leads to the same contradiction as in the previous case.
If the incline at $t_1$ was $d$, we distinguish two cases:
Either $t_1$ was inserted into the queue before $t_2$.
This is a contradiction because the cost at $t_2$ was not improved and $t_2$ would not have been inserted.
If $t_2$ was inserted before $t_1$, $t_2$ would have been removed from the queue when $t_1$ was inserted.

We conclude that for one vertex, due to the assignment there are at most $b$ visits assigned to the start of a ban interval and $b$ visits assigned to the end of a ban interval.
Additionally, there is one visit for the first point of the profile and up to $(r+1)$ visits due to concave points between two consecutive visits assigned to ban intervals.
In total, we have at most $2b(r+1)+1$ visits of a vertex.
\end{proof}

\section{Visualization of Query Sets A1 and A2}
\label{sec:appendixC}

The source and destination vertices for the query sets A1 and A2 are drawn from the regions $A$ and $B$ as shown in Figure~\ref{fig:geofence_vis}.
Region $A$ is the area southeast of 49°N 4°E and northwest of 47°N 18°E.
Region $B$ is the area southeast of 46°N 4°E and northwest of 42°N 18°E.
From each region a vertex is drawn.
Queries are generated in both directions.
This setup is taken from~\cite{b-rptrc-18}.

\begin{figure}[ht]
\centering
\begin{tikzpicture}
\node[inner sep=0pt] at (0,0) {\includegraphics[width=\textwidth]{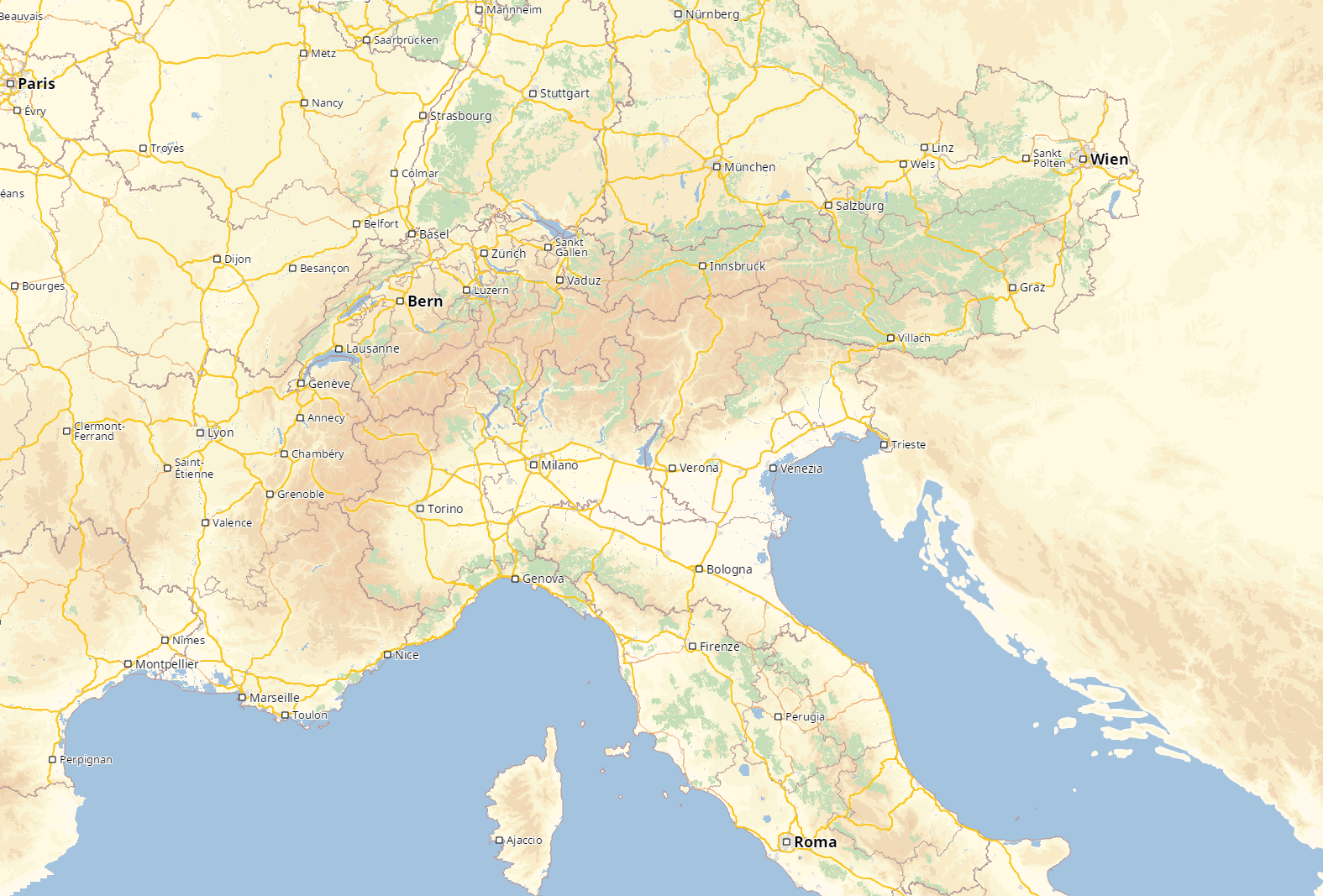}};
\draw[fill=white,opacity=0.45] (-5.95,4.42) rectangle (6.15,1.8);
\node at (0.1,3.095) {\Huge $A$};

\draw[fill=white,opacity=0.45] (-5.95,0.49) rectangle (6.15,-4.35);
\node at (0.1,-1.9) {\Huge $B$};
\end{tikzpicture}
\caption{Source and destination areas of query sets A1 and A2.}
\label{fig:geofence_vis}
\end{figure}


\begin{thebibliography}{10}

\bibitem{bdgmpsww-rptn-16}
Hannah Bast, Daniel Delling, Andrew~V. Goldberg, Matthias
  {M{\"u}ller--Hannemann}, Thomas Pajor, Peter Sanders, Dorothea Wagner, and
  Renato~F. Werneck.
\newblock {Route Planning in Transportation Networks}.
\newblock In Lasse Kliemann and Peter Sanders, editors, {\em Algorithm
  Engineering - Selected Results and Surveys}, volume 9220 of {\em Lecture
  Notes in Computer Science}, pages 19--80. Springer, 2016.
\newblock URL: \url{http://www.springer.com/gp/book/9783319494869}.

\bibitem{bgsv-mtdtt-13}
Gernot~Veit Batz, Robert Geisberger, Peter Sanders, and Christian Vetter.
\newblock {Minimum Time-Dependent Travel Times with Contraction Hierarchies}.
\newblock {\em ACM Journal of Experimental Algorithmics}, 18(1.4):1--43, April
  2013.

\bibitem{bdpw-dtdrp-16}
Moritz Baum, Julian Dibbelt, Thomas Pajor, and Dorothea Wagner.
\newblock {Dynamic Time-Dependent Route Planning in Road Networks with User
  Preferences}.
\newblock In {\em Proceedings of the 15th International Symposium on
  Experimental Algorithms (SEA'16)}, volume 9685 of {\em Lecture Notes in
  Computer Science}, pages 33--49. Springer, 2016.
\newblock URL:
  \url{http://link.springer.com/chapter/10.1007/978-3-319-38851-9_3}.

\bibitem{b-rptrc-18}
Christian Br{\"a}uer.
\newblock {Route Planning with Temporary Road Closures}.
\newblock Master's thesis, Karlsruhe Institute of Technology, 2018.

\bibitem{d-tdsr-11}
Daniel Delling.
\newblock {Time-Dependent {SHARC}-Routing}.
\newblock {\em Algorithmica}, 60(1):60--94, May 2011.
\newblock URL: \url{http://dx.doi.org/10.1007/s00453-009-9341-0}.

\bibitem{dn-crdtd-12}
Daniel Delling and Giacomo Nannicini.
\newblock {Core Routing on Dynamic Time-Dependent Road Networks}.
\newblock {\em Informs Journal on Computing}, 24(2):187--201, 2012.

\bibitem{desaulniers2000shortest}
Guy Desaulniers and Daniel Villeneuve.
\newblock The shortest path problem with time windows and linear waiting costs.
\newblock {\em Transportation Science}, 34(3):312--319, 2000.

\bibitem{dsw-cch-15}
Julian Dibbelt, Ben Strasser, and Dorothea Wagner.
\newblock {Customizable Contraction Hierarchies}.
\newblock {\em ACM Journal of Experimental Algorithmics}, 21(1):1.5:1--1.5:49,
  April 2016.
\newblock URL: \url{http://doi.acm.org/10.1145/2886843}.

\bibitem{d-ntpcg-59}
Edsger~W. Dijkstra.
\newblock {A Note on Two Problems in Connexion with Graphs}.
\newblock {\em Numerische Mathematik}, 1(1):269--271, 1959.

\bibitem{d-aassp-69}
Stuart~E. Dreyfus.
\newblock {An Appraisal of Some Shortest-Path Algorithms}.
\newblock {\em Operations Research}, 17(3):395--412, 1969.

\bibitem{gssv-erlrn-12}
Robert Geisberger, Peter Sanders, Dominik Schultes, and Christian Vetter.
\newblock {Exact Routing in Large Road Networks Using Contraction Hierarchies}.
\newblock {\em Transportation Science}, 46(3):388--404, August 2012.

\bibitem{hnr-afbhd-68}
Peter~E. Hart, Nils Nilsson, and Bertram Raphael.
\newblock {A Formal Basis for the Heuristic Determination of Minimum Cost
  Paths}.
\newblock {\em IEEE Transactions on Systems Science and Cybernetics},
  4:100--107, 1968.

\bibitem{ndls-bastd-12}
Giacomo Nannicini, Daniel Delling, Leo Liberti, and Dominik Schultes.
\newblock {Bidirectional {A*} Search on Time-Dependent Road Networks}.
\newblock {\em Networks}, 59:240--251, 2012.
\newblock Best Paper Award.

\bibitem{or-tnp-89}
Ariel Orda and Raphael Rom.
\newblock {Traveling without waiting in time-dependent networks is NP-hard}.
\newblock Technical report, Dept. Electrical Engineering, Technion-Israel
  Institute of Technology, 1989.

\bibitem{pugliese2013survey}
Luigi Di~Puglia Pugliese and Francesca Guerriero.
\newblock A survey of resource constrained shortest path problems: Exact
  solution approaches.
\newblock {\em Networks}, 62(3):183--200, 2013.

\bibitem{strasser2019perfect}
Ben Strasser and Tim Zeitz.
\newblock A* with perfect potentials, 2019.
\newblock \href {http://arxiv.org/abs/1910.12526} {\path{arXiv:1910.12526}}.

\bibitem{twb-rpbtd-18}
Marieke van~der Tuin, Mathijs de~Weerdt, and Gernot~Veit Batz.
\newblock {Route Planning with Breaks and Truck Driving Bans Using
  Time-Dependent Contraction Hierarchies}.
\newblock In {\em Proceedings of the Twenty-Eigth International Conference on
  Automated Planning and Scheduling}. AAAI Press, 2018.
\newblock URL:
  \url{https://www.semanticscholar.org/paper/Route-Planning-with-Breaks-and-Truck-Driving-Bans-Tuin-Weerdt/85c067c0a033f11166d114fcfde093d3250bb8fd}.

\end{thebibliography}
\end{document}